\documentclass[conference, letterpaper]{IEEEtran}
\usepackage{amsmath,amsfonts}
\usepackage{algorithmic}
\usepackage{array}
\usepackage[caption=false,font=footnotesize,labelfont=rm,textfont=rm]{subfig}

\usepackage{textcomp}
\usepackage{stfloats}
\usepackage{url}
\usepackage{verbatim}
\usepackage{graphicx}
\usepackage{balance}

\usepackage{amsmath}
\usepackage{makecell}
\usepackage{caption}
\usepackage{subcaption}

\usepackage{longtable}
\usepackage{lipsum}
\usepackage{tabulary}
\usepackage{graphicx}
\usepackage{subfiles}

\usepackage[open,openlevel=3,atend,numbered]{bookmark}

\usepackage[table,dvipsnames]{xcolor}
\usepackage{listings}

\usepackage{enumitem}
\usepackage{array}
\usepackage{multirow}
\usepackage{makecell}
\usepackage[export]{adjustbox}
\usepackage{amsthm,amsmath,amssymb}
\usepackage{mathrsfs}

\usepackage{orcidlink}
\newcounter{magicrownumbers}

\def\vol{{\rm vol}}

\def\Var{{\rm Var}}

\def\TV{{\rm TV}}

\usepackage{accents}
\newcommand{\ubar}[1]{\underaccent{\bar}{#1}} 
 
\newcommand{\anglebinom}[2]{\left<\genfrac{}{}{0pt}{}{#1}{#2}\right>}

\usepackage[backend=bibtex,style=ieee,sorting=none,url = false, doi = false,isbn = false]{biblatex}
\AtEveryBibitem{\clearfield{language}\clearlist{language}\clearfield{langid}\clearlist{langid}} 

\newtheoremstyle{thm_custom}
{3pt} 
{3pt} 
{\itshape} 
{} 
{\itshape} 
{:} 
{.5em} 
{\thmname{#1}\thmnumber{\textit{ #2}}\thmnote{ (#3)}}   

\newtheoremstyle{def_custom}
{3pt} 
{3pt} 
{} 
{} 
{\itshape} 
{:} 
{.5em} 
{\thmname{#1}\thmnumber{\textit{ #2}}\thmnote{ (#3)}}   

\theoremstyle{thm_custom}
\newtheorem{theorem}{Theorem}
\newtheorem{lemma}{Lemma}
\newtheorem{corollary}{Corollary}
\newtheorem{proposition}{Proposition}

\theoremstyle{def_custom}

\theoremstyle{remark}

\begin{document}
\title{New Channel Coding Lower Bounds for Noisy Permutation Channels \\
\thanks{}
}

\author{
	\IEEEauthorblockN{ Lugaoze Feng}
	\IEEEauthorblockA{the School of Electronics \\
		Peking University\\
		Beijing \\
		lgzf@stu.pku.edu.cn}
	\and
	\IEEEauthorblockN{ Xunan Li}
	\IEEEauthorblockA{National Computer Network \\
		Emergency Response Technical Team \\
		Coordination Center of China\\
		Beijing \\
		lixunan@cert.org.cn}
	\and
	\IEEEauthorblockN{ Guocheng Lv}
	\IEEEauthorblockA{the School of Electronics \\
		Peking University\\
		Beijing \\
		lv.guocheng@pku.edu.cn}
	\and
	\IEEEauthorblockN{ Ye jin}
	\IEEEauthorblockA{the School of Electronics \\
		Peking University\\
		Beijing, China \\
		jinye@pku.edu.cn}
}

\markboth
{IEEE TRANSACTIONS ON COMMUNICATIONS,~Vol.~18, No.~9, September~2020}
{IEEE TRANSACTIONS ON COMMUNICATIONS,~Vol.~18, No.~9, September~2020}

\maketitle

\begin{abstract}
	Motivated by the application of point-to-point communication networks and biological storage, we investigate new achievability bounds for noisy permutation channels with strictly positive and full-rank square matrices. Our new bounds use $\epsilon$-packing with Kullback-Leibler divergence as a metric to bound the distance between distributions and are tighter than existing bounds. Additionally, Gaussian approximations of achievability bounds are derived, and the numerical evaluation shows the precision of the approximation.
\end{abstract}


\section{Introduction}
	\IEEEPARstart{T}{he} \textit{noisy permutation channel}, consisting of a discrete memoryless channel (DMC) and a uniform random permutation block, was introduced in \cite{makur_coding_2020}, which is a point-to-point communication model that captures the out-of-order arrival of packets. This situation often occurs in applications such as communication networks and biological storage systems. Previously, several advances have been made to asymptotic bounds, including the binary channels in \cite{makur_bounds_2020}, the capacity of "full-rank" DMCs in \cite{makur_coding_2020}, and the converse bounds based on divergence covering \cite{tang_capacity_2023}. 
	
	The code lengths of suitable code in such systems are in the order of thousands or hundreds \cite{maclaren_walsh_optimal_2009}, \cite{heckel_fundamental_2017}, invalidating the asymptotic assumptions in classical information theory. We initiate the study of new channel coding bounds to extend the information-theoretic results for noisy permutation channels in finite blocklength analysis. Finite blocklength analysis and finer asymptotics are important branches of research of information theory. There has been interest in this topic since the seminal work of \cite{strassen_asymptotic_nodate}, \cite{hayashi_information_2009} and \cite{polyanskiy_channel_2010}. These works suggest that the channel coding rate in the finite blocklength regime is closely related to the \textit{information density} \cite{polyanskiy_2023_book}, i.e., the stochastic measure of the input distribution and channel noise. The second-order approximation of traditional channels involves the variance of the information density, which has been shown to approximate the channel coding rate at short blocklengths closely.
	
	In noisy permutation channels, the finer asymptotic completely differs from that in traditional channels. Specifically, since the messages are mapped to different probability distributions, the only statistical information the receiver would use from received codeword $Y^n$ is which marginal distribution $Y^n$ belongs to. This paper presents new nonasymptotic achievability bounds on the maximal rate achievable for a given blocklength and average probability of error for strictly positive square matrices with full rank, which use $\epsilon$-packing \cite{shiryayev_selected_1993,yang_information-theoretic_1999} with Kullback-Leibler (KL) divergence as a metric to lower bound the distance between each distribution and are tighter than existing bounds. In addition, for such a matrix $W$, we show that the maximal achievable size of $M$ is well approximated by
	\begin{equation}
			\log M^*(n,\epsilon)  \approx \ell \log \left( \frac{ \sqrt{n}}{-\ell\Phi^{-1}(\epsilon/|\mathcal{R}|)} \right) + \log \lambda,
	\end{equation}
	where $\xi = 2\binom{|\mathcal{Y}|}{2}$, $\ell = \mathsf{rank}(W)-1$, $\lambda$ is the volume ratio of $\Delta_{\ell}^* = \{ P_X \circ W:P_X \in \Delta_{|\mathcal{X}|-1} \}$ and $\Delta_{\ell}$. 
	By the numerical evaluation, this expression leads to tight approximations of the maximal code size achievable for blocklengths $n$ as short as $100$ in binary case. 
	
	We continue this section with the motivation and notation. Section \ref{Section_Sys_Model} sets up the system model. In Section \ref{Section_New_Bounds}, we first give a method to construct a set of divergence packing centers (message set) and bounds for the packing number. Then, we give our new achievability bound and particularize this bound to the typical DMCs. Section \ref{Section_Gaussian_Approx} studies the asymptotic behavior of the achievability using Gaussian approximation analysis and applies it to the typical DMCs. We conclude this paper in Section \ref{Section_conclusion}. 
	
\subsection{Motivation}
	First, the random deletion channel is a classical model in which the receiver will not be notified of dropped symbols. While in the erasure channel, the receiver is notified of dropped symbols. Prior work in coding theory has studied the random deletion and erasure channels \cite{diggavi2001transmission,metzner2009simplification,mitzenmacher2006polynomial}, which the error-correcting codes for these models can be as preliminary steps for noisy permutation channels \cite{gadouleau_binary_2010}. Secondly, permutation channels are a suitable model for the multipath routed network in which packets arrive in different delays (out of order) \cite{walsh_optimal_2008,maclaren_walsh_optimal_2009}. Some recent works focus on prefect codes for such networks \cite{kovacevic_subset_2013,kovacevic_perfect_2015,kovacevic_codes_2018}. Finally, DNA-based storage systems encode data as unordered codewords of nucleotides \cite{yazdi_dna-based_2015,kiah2016codes}. The random permuted DNA fragments are described by the noisy shuffling channel \cite{shomorony_capacity_2019}. These systems are also the
	motivation for studying noisy permutation channels.
	
\subsection{Notation}
We use $[n] = \{ 1,...,n \}$, $\mathbb{Z}_{\ge 0} = \{ 0, 1,... \}$ to represent integer intervals. Let $1\{ \cdot \}$ denote the indicator function. For given $\mathcal{X}$ and random variable $X \in \mathcal{X}$, we write $X \sim P_X$ to indicate that the random variable $X$ follows the distribution $P_X$. Let $X^n = (X_1,...,X_n)$ and $x^n = (x_1,...,x_n)$ denote the random vector and its realization in the $n$-th Cartesian product $\mathcal{X}^n$, respectively. A simplex on $\mathbb{R}^{|\mathcal{X}|-1}$ is a set of points 
\begin{equation}
	\Delta_{{|\mathcal{X}|-1}} = \left\{ (s_1,s_2,...,s_{|\mathcal{X}|})\in \mathbb{R}^{|\mathcal{X}|}, \ s_x \ge 0, \ \sum_{x=1}^{|\mathcal{X}|} s_x = 1 \right\}.
\end{equation}
For a $\Delta_{{|\mathcal{X}|-1}}^* \subset \Delta_{{|\mathcal{X}|-1}}$, the volume of $\Delta_{{|\mathcal{X}|-1}}^*$ is denoted by $\vol(\Delta_{{|\mathcal{X}|-1}}^*)$.

The KL divergence and the total variation are denoted by $D(\cdot\|\cdot)$ and $\TV(\cdot,\cdot)$, respectively. For a matrix $A$, we use notation $\mathsf{rank}(A)$ to represent the rank of matrix $A$. 
The probability and mathematical expectation are denoted by $\mathbb{P}[\cdot]$ and $\mathbb{E}[\cdot]$, respectively. The cumulative distribution function of the standard normal distribution is denoted by $\Phi(x)$ and $\Phi^{-1}(\cdot)$ is its inverse function.	

\section{System Model} \label{Section_Sys_Model}	
Refers to \cite{polyanskiy_2023_book}, the code $\mathcal{C}_n$ consists of a message set $\mathcal{M}$, an (possibly randomized) encoder function $f_n: \mathcal{M} \mapsto \mathcal{X}^n$, and a (possibly randomized) decoder function $g_n: \mathcal{Y}^n \mapsto \mathcal{M} \cup \{ e \}$. We write $\mathcal{X} = [|\mathcal{X}|]$ for the finite input alphabet and $\mathcal{Y} = [|\mathcal{Y}|]$ for the finite output alphabet. Denote by $|\mathcal{M}|$ the maximal code size of code $\mathcal{C}_n$.  

The sender uses encoder $f_n$ to encode the message $M$ into a $n$-length codeword $X^n$ and passes through the discrete memoryless channel $W$ to produce $Y^n \in \mathcal{Y}^n$. Then, the codeword $Y^n$ goes through the random permutation block and result in $Y_{\rm Perm}^n \in \mathcal{Y}^n$. Finally, the receiver uses decoder $g_n$ to estimate the message $\hat{M}$. We can describe these steps by the following Markov chain:
\begin{equation}
	M \rightarrow X^n \rightarrow Y^n \rightarrow Y_{\rm Perm}^n \rightarrow \hat{M}.
\end{equation}

According to \cite[Lemma 2]{makur_coding_2020}, we can describe the Markov chain of noisy permutation channels as the following equavelent model:
\begin{equation}
		M \rightarrow Z^n \rightarrow X^n \rightarrow Y^n \rightarrow \hat{M},
\end{equation}
where $Z^n$ is the original codeword, and $X^n$ is its random permutation. Both of them are taking value in $\mathcal{X}$. $X^n$ passes through the probability kernel $W$ of DMC and becomes the sequence $Y^n$, where each $Y_i \in \mathcal{Y}$. The random permutation block $P_{X^n|Z^n}$ operates as follows. First denote a random permutation as $\sigma: \{ 1,...,n \} \mapsto \{ 1,...,n \}$, drawn randomly and uniformly from the symmetric group $\mathcal{S}_n$ over $\{ 1,...,n \}$. Then we have $X_{\sigma(i)} = Z_i$ for all $i \in \{ 1,...,n \}$. The random permutation block does not change the probability distribution of codewords \cite{makur_coding_2020}, i.e., if $Z^n \stackrel{\rm i.i.d.}{\sim} P_X$, then $X^n \stackrel{\rm i.i.d.}{\sim} P_X$.

The DMC is defined by a $|\mathcal{X}| \times |\mathcal{Y}|$ matrix $W$, where $W(y|x)$ denotes the probability that the output $y \in \mathcal{Y}$ occurs given input $x \in \mathcal{X}$. We say $W $ is strictly positive if all the transition probabilities in $W $ are greater than $0$. In this paper, we assume $W$ is a strictly positive and full rank square matrix, i.e., $|\mathcal{X}| = |\mathcal{Y}| = \mathsf{rank}(W)$.

 For a given code $\mathcal{C}_n = (\mathcal{M},f_n,g_n)$, the average error probability of code $\mathcal{C}_n$ is
\begin{align}
	P_e: & = \mathbb{P} [M \not = \hat{M}].
\end{align}

The maximal code size achievable with a given blocklength and probability of error is denoted by $M^*(n,\epsilon) = \max \{|\mathcal{M}|: \exists \ \mathcal{C}_n \ {\rm such \ that }\ P_e \le \epsilon \}.$
The code rate with the given encoder-decoder pair $(f_n, g_n)$ is denoted by
\begin{equation}\label{eq:rate_function}
	R(n,\epsilon) \stackrel{\triangle}{=} \frac{\log M^*(n,\epsilon)}{\log n},
\end{equation}
where $\log(\cdot)$ is the binary logarithm (with base 2) in this paper. 
The capacity for noisy permutation channels is defined as
\begin{equation} 
	C \stackrel{\triangle}{=} \lim \limits_{\epsilon \rightarrow 0^+} \liminf \limits_{n \rightarrow \infty} \frac{\log M^*(n,\epsilon)}{\log n}.
\end{equation}

\section{New Bounds on Rate} \label{Section_New_Bounds}
In this section, we introduce our new bounds. We first constructed the message set using \textit{divergence pakcing}, i.e., $\epsilon$-packing with KL divergence as distance. Then, we give our new achievability bound. The key ingredient is our analysis of error events. Finally, using the message set, we particularize this new bound to BSC and BEC permutation channels.

\subsection{Message set and Divergence Packing} \label{Section_Divergence_Pcking}

A divergence pakcing is a set of centers in simplex $\Delta_{|\mathcal{Y}|-1}$ such that the minimum distance between each center is greater than some KL distance. Since the messages correspond to different distributions in noisy permutation channels, the message set can be equivalent to the set of marginal distributions at the receiver. The non-asymptotic channel performance is related to the likelihood ratio of two distributions (decoding metric) when using the ML decoder. The statistical mean of this decoding metric is the KL divergence, which arises from applying the law of large numbers as the blocklength grows. Thus, divergence packing can obtain an upper bound of the error probability by lower bounding the distance between each distribution. This motivates us to use it in constructing the message set.

Let $\mathcal{N}_{r_0,|\mathcal{Y}|} = \{P_1,...,P_{|\mathcal{M}|}\} \subset \Delta_{ |\mathcal{Y}|-1}$ be the set of divergence packing centers. The divergence packing number is defined as 
\begin{align}
	|\mathcal{N}_{r_0,|\mathcal{Y}|}| = \max \{ |\mathcal{M}|: & \exists \{P_1,...,P_{|\mathcal{M}|}\} \nonumber \\ 
	&\ {\rm s.t.} \ \min \limits_{i \neq j} D(P_i\|P_j) \ge r_0 \}.
\end{align}
Denote by $\Delta_{|\mathcal{Y}|-1}^* = \{ P_X \circ W:P_X \in \Delta_{|\mathcal{X}|-1} \}$ the achievability space of the marginal distribution. Similarly, denote by $\mathcal{N}_{r_0,|\mathcal{Y}|}^*$ the set of packing centers on $\Delta_{ |\mathcal{Y}|-1}^*$. In the sequel, we denote the marginal distribution corresponding to message $m$ by $P_m \in \mathcal{N}_{r_0,|\mathcal{Y}|}^*$.

Next, we introduce two methods for constructing the set of packing centers (message set) and the bounds of their size. Let $r>0$, and we consider the following set:
\begin{align} \label{eq:Gamma_message_set_DMCs}
	\Gamma_{r,|\mathcal{Y}|}^* = \bigg \{  P \in \Delta_{|\mathcal{Y}|-1}^*: & P = \left( \frac{a_1}{\lfloor 1/r \rfloor},...,\frac{a_{|\mathcal{Y}|}}{\lfloor 1/r \rfloor}\right) \nonumber \\
	& \ \ {\rm where} \ a_1,...,a_{|\mathcal{Y}|} \in \mathbb{Z}_{\ge 0}  \bigg \}.
\end{align} 
Then, we have the following result established in Appendix \ref{Appendix_divergence_packing}.

\begin{theorem} \label{thm:packing_num_subs}
	Fix a $\Delta_{|\mathcal{Y}|-1}^* \subset \Delta_{|\mathcal{Y}|-1}$ with dimension $|\mathcal{Y}|-1$. We can construct $\mathcal{N}_{r_0,|\mathcal{Y}|}^*$ by (\ref{eq:Gamma_message_set_DMCs}) with $r = \sqrt{\frac{r_0}{2 \log e}}$. Then, we have
	\begin{align}
		|\mathcal{N}_{r_0,|\mathcal{Y}|}^*| \ge & \lambda \left( \frac{1/r + |\mathcal{Y}| - 2}{|\mathcal{Y}|-1} \right)^{|\mathcal{Y}|-1} \nonumber \\
		 & \ \ \ \ \  - 2 |\mathcal{Y}| \left( \frac{( 1/r + |\mathcal{Y}| - 2)e}{|\mathcal{Y}| - 2} \right)^{|\mathcal{Y}| - 2},
	\end{align}
	where
	\begin{equation} \label{eq:definition_volume_ratio}
		\lambda = \frac{\vol(\Delta_{|\mathcal{Y}|-1}^*)}{\vol(\Delta_{|\mathcal{Y}|-1})}.
	\end{equation}
\end{theorem}

In the binary case, we consider $\Delta_{1}^* = \{ (q,1-q): \delta_1 \le q \le 1-\delta_2 \}$, where $\delta_1 > 0$ and $ \delta_2>0$. Define
\begin{align}\label{eq:Gamma_message_set_binary}
	\Gamma_{r,2}^* = \bigg \{  (q,1-q): & q =  \frac{\xi a_1}{\lfloor 1/r \rfloor} + \delta_1,\ {\rm where} \ a_1 \in \mathbb{Z}_{\ge 0}  \bigg \},
\end{align}
where $\xi = 1-\delta_1-\delta_2$.
Then, we have the following result proved in Appendix \ref{Appendix_divergence_packing}. Compared to Theorem \ref{thm:packing_num_subs}, it has no loss in the constant in binary case.

\begin{proposition} \label{proposition:packing_num_binary}
	Fix $\Delta_{1}^* = \{ (q,1-q): \delta_1 \le q \le 1-\delta_2 \}$, where $\delta_1 > 0$ and $ \delta_2>0$. We can construct $\mathcal{N}_{r_0,2}^*$ by (\ref{eq:Gamma_message_set_binary}) with $r = \frac{1}{\xi} \sqrt{\frac{r_0}{2 \log e}}$ and $\xi = 1-\delta_1-\delta_2$ such that
	\begin{equation}
		|\mathcal{N}_{r_0,2}^*| = \lfloor 1/r \rfloor +1.
	\end{equation}
\end{proposition}

\subsection{Achievability Bound}

This subsection develops a new achievability bound for noisy permutation channels. Our new achievability bound is based on random coding and divergence packing, which yields the spirit of random coding union (RCU) bound \cite[Theorem 16]{polyanskiy_channel_2010}. 

We first introduce some definitions. Fix a $P \in \mathcal{N}_{r_0,|\mathcal{Y}|}^*$, we are often concerned with the divergence packing centers close to $P$. To do this, we consider $Q_{m,n}$, where $m \neq n$, $q_m = p_m + \frac{K}{\lfloor 1/r \rfloor}, q_n = p_n - \frac{K}{\lfloor 1/r \rfloor}$ and $q_i = p_i$ for $i \in \mathcal{Y} \setminus \{ m,n \}$. $K$ is a constant that has value $1$ or $\xi$ when $\mathcal{N}_{r_0,|\mathcal{Y}|}^*$ is constructed by (\ref{eq:Gamma_message_set_DMCs}) or (\ref{eq:Gamma_message_set_binary}), respectively. Consider the distributions belonging to
\begin{equation}
	\mathcal{R}_{P} = \left\{ Q_{m,n}: m,n \in \mathcal{Y}, m\neq n \right\} \cap \mathcal{N}_{r_0,|\mathcal{Y}|}^*.
\end{equation}
For convenience, we use $j \in [|\mathcal{R}_{P}|]$ to index $ Q_j \in \mathcal{R}_{P}$. By counting, we have $|\mathcal{R}_{P}| \le 2 \binom{|\mathcal{Y}|}{2}$.

For the marginal distribution $P_m$ corresponding to the transmitted message $m$, we use the log-likelihood ratio to define the following decoding metric:
\begin{equation}
	d(m,j,y): = \log \frac{P_{m}(y)}{Q_j(y)},
\end{equation}
where $Q_j \in \mathcal{R}_{P_m}$. The following lemma proved in Appendix \ref{Appendix_lemma_event_disjoint} is useful.
\begin{lemma}\label{lemma:event_disjoint}
	For the set of marginal distributions constructed by (\ref{eq:Gamma_message_set_DMCs}) or (\ref{eq:Gamma_message_set_binary}), we have 
	\begin{align} \label{eq:lemma:event_disjoint_high_dim}
		&\mathbb{P} \left[ \bigcup_{j=1,j\neq m}^{|\mathcal{M}|}\left\{  P_{j}^n \ge P_{m}^n \right\} \right] = \mathbb{P} \left[ \bigcup_{j =1}^{|\mathcal{R}_{P_m}|} \left\{ Q^n_j \ge P_{m}^n \right\} \right],
	\end{align}
	where the sequence $Y^n$ is independent identically distributed (iid) from $P_{m}$.
\end{lemma}	
If the transmitted message is $m$, Lemma \ref{lemma:event_disjoint} shows that the union of events $\cup_{j=1,j\neq m}^{|\mathcal{M}|}\left\{  P_{j}^n \ge P_{m}^n \right\}$ can be equivalent to a minor union $\cup_{j =1}^{|\mathcal{R}_{P_m}|}  \left\{ Q^n_j \ge P_{m}^n \right\}$ on $\mathcal{R}_{P_m}$. Its size only depends on the size of the output alphabet.

Now, we are ready to present our new achievability bound.
\begin{theorem}\label{thm:random_coding_bound}
	Fix a strictly positive and full-rank square matrix $W$ for noisy permutation channels. Let the set of marginal distributions $\mathcal{N}_{r_0,|\mathcal{Y}|}^*$ be constructed by (\ref{eq:Gamma_message_set_DMCs}) or (\ref{eq:Gamma_message_set_binary}) and let
	\begin{equation}  \label{eq:thm:random_coding_bound_2}
		L_n(m) = \sum_{j=1}^{|\mathcal{R}_{P_m}|}\mathbb{P}_m \left[ \sum_{i=1}^{n}d(m,j,Y_i) \le 0  \right].
	\end{equation} 
	Conditioning on $M=m$, the sequence $Y^n$ is iid from $P_m$.
	Then, there exists a code $\mathcal{C}_n$ (average error probability) with size $|\mathcal{N}_{r_0,|\mathcal{Y}|}^*|$ such that
	\begin{equation} \label{eq:thm:random_coding_bound_1}
		\epsilon \le \min \left\{ 1,\mathbb{E} \left[ L_n(M) \right]\right\},
	\end{equation}
	where $M$ is uniform on $\mathcal{M}$.
\end{theorem}
\begin{proof}
	 For $P_X \in \Delta_{|\mathcal{X}|-1}$, the marginal distribution is the linear combination of $P_X(\cdot) $. Since the matrix is full-rank, we obtain the achievability space of the marginal distribution is a $(|\mathcal{Y}|-1)$-dimensional probability space $\Delta_{|\mathcal{Y}|-1}^*$. We consider constructing the set of marginal distributions by using the divergence packing on space $\Delta_{|\mathcal{Y}|-1}^*$, i.e., let $\mathcal{N}_{r_0,|\mathcal{Y}|}^*$ be the set of marginal distributions. 
	
	Assuming the transmitted message is $m \in \mathcal{M}$, corresponding to the marginal distribution $P_m$.
	Then, the maximum likelihood decoder shows that the average probability of error satisfies
	\begin{align} \label{eq:thm:random_coding_bound_proof_1}
		\mathbb{P}[\text{error}|M=m] 
		& \le \sum_{y^n \in \mathcal{Y}^n}  P_{m}^n(y^n) \nonumber \\
		& \ \ \ \  \ \ \times  \bigcup_{j=1,j\neq m}^{|\mathcal{M}|} 1\left\{  P_{j}^n(y^n) \ge P_{m}^n(y^n) \right\} \nonumber \\
		& = \mathbb{P}_m\left[ \bigcup_{j=1,j \neq m}^{|\mathcal{M}|}\left\{  P_{j}^n \ge P_{m}^n \right\} \right].
	\end{align}
	This is an inequality since we regard the equality case, $P_{j}^n(Y^n) \ge P_{m}^n(Y^n)$, as an error event though the ML decoder might return the correct message.
	By using Lemma \ref{lemma:event_disjoint}, we continue the bounding as follows:
	\begin{align}
		\mathbb{P}[\text{error}|M=m] 
		& \le \mathbb{P}_m \left[ \bigcup_{j =1}^{|\mathcal{R}_{P_m}|} \{ Q_j^n \ge P_{m}^n \} \right] \nonumber \\
		& \le \sum_{j=1}^{|\mathcal{R}_{P_m}|} \mathbb{P}_m \left[\sum_{i=1}^{n}d(m,j,Y_i) \le 0 \right].
	\end{align}
	Let the message is uniform on $\mathcal{M}$. Then, we take the expectation overall codewords to obtain
	\begin{align}
		\epsilon 
		& \le \frac{1}{|\mathcal{M}|} \sum_{m=1}^{|\mathcal{M}|} \sum_{j=1}^{|\mathcal{R}_{P_m}|} \mathbb{P}_{m} \left[ \sum_{i=1}^{n}d(m,j,Y_i) \le 0 \right]. \nonumber \\
		& = \mathbb{E} \left[ L_n(M) \right].
	\end{align}
	This completes the proof.
\end{proof}

	In the view of Theorem \ref{thm:random_coding_bound}, some remarks are in order. Firstly, we do not use the union bound for scaling but analyze the relationship between error events by Lemma \ref{lemma:event_disjoint}. Thus, for each $P_m$, Theorem \ref{thm:random_coding_bound} upper bounds the probability of error with the sum of the probabilities of the events on $\mathcal{R}_{P_m}$, which makes our bound much stronger than \cite{makur_coding_2020}. Additionally, this conclusion is consistent with our intuition: the rate of decay of $P_e$ is dominated by the rate of decay of the probability of error in distinguishing two "close" messages.

	Secondly, Theorem \ref{thm:random_coding_bound} relies on the message set constructed by (\ref{eq:Gamma_message_set_DMCs}) or (\ref{eq:Gamma_message_set_binary}). We restrict the channel $W$ to be a full-rank square matrix, which makes $\Delta_{|\mathcal{Y}|-1}^*$ an equal-dimensional subspace of $\Delta_{|\mathcal{Y}|-1}$ and the evenly spaced grid structure on $\Delta_{|\mathcal{Y}|-1}^*$ can be constructed. Otherwise, we cannot apply (\ref{eq:Gamma_message_set_DMCs}) unless we make strong assumptions about $W$. If the DMC is a non-strictly positive matrix, we need to make additional assumptions on $\Delta_{|\mathcal{Y}|-1}^*$. For example, take a subset of $\Delta_{ |\mathcal{Y}|-1}^*$ such that for any $P \in \Delta_{ |\mathcal{Y}|-1}^*$, we have $P(y) \neq 0$.

\subsection{BSC Permutation Channels}
In this subsection, we particularize the nonasymptotic bounds to the BSC, i.e., the DMC matrix is
\begin{equation}
	W = 
	\begin{bmatrix}
		1-\delta &\delta \\
		\delta & 1-\delta 
	\end{bmatrix},
\end{equation}
denoted $\mathsf{BSC}_{\delta}$.
For convenience, we denote $P_{m}(\cdot) = (\delta_m,1-\delta_m)$. For $i,j\in [|\mathcal{M}|]$, let $\delta_i < \delta_j$ if $i < j$. Then, for a $P_m$, we clearly have
\begin{equation}
	\mathcal{R}_{P_m} = 
	\begin{cases}
		\{ P_{m-1},P_{m+1} \}, & 2\le m \le |\mathcal{M}|-1,\\
		\{ P_{2} \}, &  m = 1 ,\\
		\{ P_{|\mathcal{M}|-1} \}, & m = |\mathcal{M}|.\\
	\end{cases}
\end{equation}
Let
\begin{equation}
	\anglebinom{n}{\ubar{T}_{i}} =
	\begin{cases}
		\sum_{t =0}^{\ubar{T}_{i}}  \binom{n}{t} \delta_i^t(1-\delta_i)^{n-t},& i \ge 2, \\
		0,& i = 1.
	\end{cases} 
\end{equation}
and
\begin{equation}
	\anglebinom{n}{\bar{T}_{i}} =
	\begin{cases}
		\sum_{t =\bar{T}_{i}}^{n}  \binom{n}{t} \delta_i^t(1-\delta_i)^{n-t},& i \le |\mathcal{M}|-1, \\
		0,& i = |\mathcal{M}|.
	\end{cases}
\end{equation}
The following bound is a straightforward generalization of Theorem \ref{thm:random_coding_bound}. 
\begin{theorem}[Achievability]\label{thm:random_coding_bound_BSC_Perm}
	For the BSC permutation channel with crossover probability $\delta$, there exists a code $\mathcal{C}_n$ suth that
	\begin{align} \label{eq:thm:random_coding_bound_BSC_Perm_1}
		\epsilon \le \frac{\sum_{i=1}^{|\mathcal{M}|}}{|\mathcal{M}|} \min \bigg \{ 1, \anglebinom{n}{\ubar{T}_{i}} + \anglebinom{n}{\bar{T}_{i}}  \bigg \},
	\end{align}
	where
	\begin{equation}
		\ubar{T}_{i} = \Bigg \lfloor \frac{n\log \frac{1-\delta_{i-1}}{1-\delta_i}}{\log \frac{\delta_i(1-\delta_{i-1})}{\delta_{i-1}(1-\delta_i)}} \Bigg \rfloor
	\end{equation}
	and
	\begin{equation}
		\bar{T}_{i} = \Bigg \lceil \frac{n\log \frac{1-\delta_{i+1}}{1-\delta_i}}{\log \frac{\delta_i(1-\delta_{i+1})}{\delta_{i+1}(1-\delta_i)}} \Bigg \rceil.
	\end{equation}
	The set of marginal distributions is constructed by (\ref{eq:Gamma_message_set_binary}) and for the radius $r$, we have
	\begin{equation} \label{eq:thm:random_coding_bound_BSC_Perm_2}
		|\mathcal{M}|  = \lfloor 1/r \rfloor + 1.
	\end{equation}
\end{theorem}
\begin{proof}
	Assuming the transmitted message is $m \in \mathcal{M}$, corresponding to the marginal distribution $P_m$. The term corresponding to $P_{m-1}$ in (\ref{eq:thm:random_coding_bound_2}) can be computed as
	\begin{align} \label{eq:thm:random_coding_bound_BSC_Perm_proof_1}
		\mathbb{P}_m\left[ \sum_{i=1}^{n}d(m,m-1,Y_i) \le 0 \right] = \sum_{t=0}^{\ubar{T}_{m}} \binom{n}{t} \delta_m^t(1-\delta_m)^{n-t},
	\end{align}
	where $\ubar{T}_{m}$ follows from (\ref{eq:binary_case_event_compute}) in Appendix \ref{Appendix_lemma_event_disjoint}. 
	Similarly, the term corresponding to $P_{m+1}$ in (\ref{eq:thm:random_coding_bound_2}) can be computed as
	\begin{align} \label{eq:thm:random_coding_bound_BSC_Perm_proof_2}
		\mathbb{P}_m\left[ \sum_{i=1}^{n}d(m,m+1,Y_i) \le 0 \right] = \sum_{t=\bar{T}_{m}}^{n} \binom{n}{t} \delta_m^t(1-\delta_m)^{n-t},
	\end{align}
	where $\bar{T}_{m}$ follows from (\ref{eq:binary_case_event_compute}). 
	Substitute (\ref{eq:thm:random_coding_bound_BSC_Perm_proof_1}) and (\ref{eq:thm:random_coding_bound_BSC_Perm_proof_2}) into Theorem \ref{thm:random_coding_bound} to complete the proof.
\end{proof}

\subsection{BEC Permutation Channels}
	The BEC permutation channel consists of input alphabet $\mathcal{X} = \{ 0,1 \}$ and ouput alphabet $\mathcal{Y} = \{0,e,1\}$, where the conditional distributions is
	\begin{equation}
		\forall y \in \mathcal{Y}, \forall x \in \mathcal{X}, W(y|x) = 
		\begin{cases}
			1- \eta, & y=x \\
			\eta, & y=e \\
			0,& {\rm otherwise}
		\end{cases}
	\end{equation}
	Moreover, we denote such a channel $W$ as $\mathsf{BEC}_{\eta}$ for convenience.
	
	Next, we have the following achievability bound.
	\begin{theorem} \label{thm:random_coding_bound_BEC}
		For BEC permutation channels with erasure probability $2\delta$, there exists a code $\mathcal{C}_n$ such that the average probability of error and the maximal code size satisfy (\ref{eq:thm:random_coding_bound_BSC_Perm_1}) and (\ref{eq:thm:random_coding_bound_BSC_Perm_2}), respectively.
	\end{theorem}
	\begin{proof}
		The derivations of this proof follow from \cite[Proposition 6]{makur_coding_2020}, and we include the details for the sake of completeness. We first note that the BSC matrix satisfies the Doeblin minorization condition (e.g., see \cite[Definition 5]{makur_coding_2020}) with $\mathbf{1}^\top/2$ and constant $2\delta$. Using \cite[Lemma 6]{makur_coding_2020} to obtain that $\mathsf{BSC}_{\delta}$ is a degraded version of $\mathsf{BEC}_{2\delta}$. Then, For the encoder and decoder pairs $(f_n,g_n)$ for BSC permutation channels and $(f_n,\tilde{g}_n)$ for BSC permutation channels, the average probability of error satisfies \cite[Eq. (70)]{makur_coding_2020}
		\begin{equation}
			P_e (f_n,g_n,\mathsf{BSC}_{\delta}) = P_e(f_n,\tilde{g}_n,\mathsf{BEC}_{2\delta}).
		\end{equation}
		This completes the proof.
	\end{proof}

\section{Gaussian Approximation} \label{Section_Gaussian_Approx}
	We turn to the asymptotic analysis of the noisy permutation channel for a given blocklength and average probability of error. 

\subsection{main result}
The main result in this section is the following.
\begin{theorem} \label{thm:Approximation_of_DMCs}
	Fix $W$ is a strictly positive and full-rank square matrix for noisy permutation channels. Then,  there exists a number $N_0 \ge 1$, such that 
	\begin{equation}
		\log M^*(n,\epsilon) = \ell \log \left( \frac{ \sqrt{n}}{-\ell\Phi^{-1}(\epsilon/|\mathcal{R}|)} \right) + \log \lambda + \theta
	\end{equation}
	is achievable for all $n \ge N_0$,
	where $|\mathcal{R}| = 2\binom{|\mathcal{Y}|}{2}$, $\ell = |\mathcal{Y}|-1$, $\lambda$ is the volume ratio of $\Delta_{\ell}^*$ and $\Delta_{\ell}$, and $\theta$ is a constant.
\end{theorem}

To develop the Gaussian Approximation, we start with some useful definitions. The variance and third absolute moment of log-likelihood ratio between two distributions $P$ and $Q$ are defined as $
V(P\|Q) = \mathbb{E}_{P}\left[ \left( \log \frac{P}{Q} - D(P\|Q) \right)^2 \right]$ and $
T(P\|Q) = \mathbb{E}_{P}\left[ \left| \log \frac{P}{Q} - D(P\|Q) \right|^3 \right]$, respectively.

Then, we state some auxiliary results.
The first lemma is concerned with properties of $V(P\|Q)$ and $T(P\|Q)$, which is proved in Appendix \ref{Appendix_lemma_properities_Vn_Tn}.
\begin{lemma} \label{lemma:properities_Vn_Tn}
	Fix a strictly positive and full-rank square matrix $W$. Let $\mathcal{N}_{r_0,|\mathcal{Y}|}^*$ constructed by (\ref{eq:Gamma_message_set_DMCs}) be the set of packing centers on $\Delta_{|\mathcal{Y}|-1}^*$. If packing radius $r_0 \le \frac{2 \log e}{9}$, for $P \in \mathcal{N}_{r_0,|\mathcal{Y}|}^*$ and $Q \in \mathcal{R}_{P}$, we have
	\begin{equation}\label{eq:lemma:properities_Vn_Tn_1}
		V(P\|Q) = r_0 F_0, 
	\end{equation}
	where 
	\begin{equation}
		\left( \frac{5}{8 p_{max}}-\frac{2}{9}\right)\log e \le F_0 \le \frac{5 \log e}{2p_{min}(1-p_{max})^2},
	\end{equation}
	$p_{min} = \min_{y \in \mathcal{Y}} \{ P(y):P \in \Delta_{|\mathcal{Y}|-1}^* \} $, and $p_{max} = \max_{y \in \mathcal{Y}} \{ P(y):P \in \Delta_{|\mathcal{Y}|-1}^* \} $.
	Additionally, we have
	\begin{equation}\label{eq:lemma:properities_Vn_Tn_2}
		T(P\|Q) \le  \frac{36\sqrt{2} (\log e )^{3/2}}{p_{min}^2 (1-p_{max})^3}  r_0^{3/2}.
	\end{equation}
\end{lemma}

Then, we give an important tool in the Gaussian approximation analysis. 
\begin{theorem}[\textit{Berry-Esseen, \texorpdfstring{\cite[Chapter XVI.5, Theorem 2]{Feller_book}}{}}] \label{Berry-Esseen Bound}
	Fix a positive integer $n$. Let $Z_i$ indexed by $(1,...,n)$ be independent. Then, for any real $x$ and $C_0 \le 6$ we have
	\begin{equation}
		\left| \mathbb{P}\left[ \sum_{i=1}^{n}Z_i < n\left(\mu_n + x\sqrt{\frac{V_n}{n}}\right) \right] - \Phi(x) \right| \le \frac{B_n}{\sqrt{n}},
	\end{equation}
	where
	\begin{align}
		\mu_n  = \frac{1}{n} \sum_{i=1}^{n} \mathbb{E}[Z_i], \
		V_n  = \frac{1}{n} \sum_{i=1}^{n} \Var[Z_i], \\
		T_n  = \frac{1}{n} \sum_{i=1}^{n} \mathbb{E}\left[ \left| Z_i - \mu_i \right|^3 \right], \
		B_n  = C_0\frac{T_n}{V_n^{3/2}}.
	\end{align}
\end{theorem}

{\noindent \textit{Proof of Theorem \ref{thm:Approximation_of_DMCs}.} }

For the transmitted message $m$, passing codeword through the random permutation block and the channel $W^n$ induces a marginal distribution $P_m$ on $Y^n$. Note that the decoding metric is the sum of independent identically distributed variables:
\begin{equation}
	\sum_{k=1}^{n}d(m,j,Y_k) = \sum_{k=1}^{n} \log \frac{P_{m}(Y_k)}{Q_{j}(Y_k)}.
\end{equation}
It has the mean $D(P_{m}\|Q_j)$, variance $V(P_{m}\|Q_j)$, and third absolute moment $T(P_{m}\|Q_j)$. 
Denote 
\begin{equation}
	B_{m} = \max \limits_{j \in [|\mathcal{R}_{P_m}|]}  	\frac{6T(P_{m}\|Q_j)}{\sqrt{V^3(P_{m}\|Q_j)}}
\end{equation} 
and
\begin{equation}
	\mathbb{P}_e[m] = \sum_{j=1}^{|\mathcal{R}_{P_m}|} \mathbb{P}_{m}\left[ \sum_{k=1}^{n}d(m,j,Y_k) \le 0 \right].
\end{equation}
According to Theorem \ref{thm:random_coding_bound}, there exists a code such that 
\begin{equation}
	\epsilon' \le \frac{1}{|\mathcal{M}|}\sum_{i=1}^{|\mathcal{M}|} \mathbb{P}_e[m].
\end{equation}
Let the message set be $\mathcal{N}_{r_m,|\mathcal{Y}|}^*$, constructed by (\ref{eq:Gamma_message_set_DMCs}). Then, we have $\max \limits_{j \in [|\mathcal{R}_{P_m}|]} D(P_m\|Q_j) \le r_m$. Continue as follows:
\begin{align}
	\mathbb{P}_e[m] 
	& \le \frac{|\mathcal{R}_{P_m}|B_{m}}{\sqrt{n}} + \max \limits_{j \in [|\mathcal{R}_{P_m}|]}  |\mathcal{R}_{P_m}|\Phi\left( \frac{-nD(P_m\|Q_j)}{\sqrt{nV(P_m\|Q_j)}} \right) \label{eq:Approximation_of_DMCs_proof_achiev_1} \\
	&  \le \frac{|\mathcal{R}_{P_m}|B_{m}}{\sqrt{n}} + |\mathcal{R}_{P_m}| \Phi\left( -\sqrt{\frac{n r_m}{ F_0 }  } \right), \label{eq:Approximation_of_DMCs_proof_achiev_2}
\end{align}
where (\ref{eq:Approximation_of_DMCs_proof_achiev_1}) follows from Theorem \ref{Berry-Esseen Bound}, and (\ref{eq:Approximation_of_DMCs_proof_achiev_2}) holds for a suitable constant $F_0>0$ by Lemma \ref{lemma:properities_Vn_Tn}.

Now choose
\begin{equation}
	r_m = \frac{\left( \Phi^{-1}\left( \frac{\epsilon}{|\mathcal{R}_{P_m}|} - \frac{F_1}{\sqrt{n}} 	\right) \right)^2 F_0 }{n}
\end{equation}
ensuring that $\epsilon' \le \epsilon$, where $F_0>0$ and $F_1>0$ are suitable constants. This can be done since we can find some constants $F_2>0$ and $F_3>0$ such that $T(P_m\|Q_j) \le F_2 r_0^{3/2}$ and $V(P_m\|Q_j) \ge F_3 r_0$ for $n$ sufficiently large by Lemma \ref{lemma:properities_Vn_Tn}, respectively. Then, $B_{m}$ can be upper bounded by a suitable constant $F_1$. 

Note that this argument holds for all $m \in \mathcal{M}$ and we have $|\mathcal{R}_{P_m}| \le 2 \binom{|\mathcal{Y}|}{2}$. Let $\ell = |\mathcal{Y}|-1$ and let $|\mathcal{R}| = 2 \binom{|\mathcal{Y}|}{2}$. For $n$ sufficiently large, we have
\begin{align}
	\log M^*(n,\epsilon)  & = \log |\mathcal{N}_{r_0,|\mathcal{Y}|}^*| \nonumber\\
	& \ge \log \bigg (\lambda \bigg (  \frac{\sqrt{n}F_4}{-\ell\Phi^{-1}(\epsilon/|\mathcal{R}|)} + \frac{\ell -1}{\ell}  \bigg )^\ell  \nonumber \\
	 &  \ \ \ \ \ \ \ \ \ \ - F_5(\sqrt{n})^{\ell - 1} \bigg )  \label{eq:Approximation_of_DMCs_proof_achiev_3}\\
	& \ge \ell \log \left( \frac{\sqrt{n}}{-\ell\Phi^{-1}(\epsilon/|\mathcal{R}|)} \right)  + \log \lambda + \log F_6 \label{eq:Approximation_of_DMCs_proof_achiev_4},
\end{align}
where $F_6 > 0$ is a constant, (\ref{eq:Approximation_of_DMCs_proof_achiev_3}) holds for suitable $F_4 >0$ and $F_5 > 0$ by Theorem \ref{thm:packing_num_subs} and Taylor’s formula of $\Phi^{-1}(\cdot)$, (\ref{eq:Approximation_of_DMCs_proof_achiev_4}) holds because the inequality $\log (G_1x^k) - \log (G_1x^{k}-x^{k-1}) \le G_0$ holds for a constant $G_0>0$ if constant $G_1>0$ and $x$ is sufficiently large. This completes the proof.
{\hfill $\square$}

\subsection{Approximation of BSC and BEC Permutation Channels}

We apply Theorem \ref{thm:Approximation_of_DMCs} to obtain the following approximation.
\begin{corollary}\label{corollary:Gaussian_Approximation_of_BSCPermC_weak}
	For BSC permutation channels with crossover probability $\delta$, there exists a number $N_0 \ge 1$, such that 
	\begin{equation}\label{eq:Gaussian_Approximation_of_BSCPermC_weak}
		\log M^*(n,\epsilon)  = \log \left(  \frac{(1-2 \delta)\sqrt{n}}{-\Phi^{-1}\left(\frac{\epsilon}{2}\right)} \right)  + \theta
	\end{equation}
	is achievable  for all $n \ge N_0$, where $\theta$ is a constant.
\end{corollary}
\begin{proof}
	For BSC, we have $\Delta_{1}^* = \{ (p,1-p): \delta \le p \le 1-\delta\}$. By using  Lagrange’s formula \cite{stein_note_1966},
	the volume ratio of $\Delta_{1}^*$ and $\Delta_{1}$ is $1-2 \delta$. Using Theorem \ref{thm:Approximation_of_DMCs} to obtain the result.
\end{proof}

 The approximation of BSC permutation channels can also be derived from Proposition \ref{proposition:packing_num_binary}.  To use the message set constructed by (\ref{eq:Gamma_message_set_binary}), we need the following lemma:

\begin{lemma} \label{lemma:properities_Vn_Tn_binary}
	In the notation of Lemma \ref{lemma:properities_Vn_Tn}, let $\mathcal{N}_{r_0,2}^*$ constructed by (\ref{eq:Gamma_message_set_binary}) be the set of packing centers on $\Delta_{1}^*$. Then, there exists a packing radius $r_1$, such that for all $r_0\le r_1$, we have
	\begin{equation}
		F_0 r_0 \le V(P\|Q) \le F_1 r_0
	\end{equation}
	and 
	\begin{equation}
		T(P\|Q) \le F_2 r_0^{3/2},
	\end{equation}
	where $F_0$, $F_1$ and $F_2$ are positive and finite.
\end{lemma}

Then, we have the following result.
\begin{theorem}\label{thm:Gaussian_Approximation_of_BSCPermC}
	For BSC permutation channels with crossover probability $\delta$, there exists a number $N_0 \ge 1$, such that
	\begin{equation}\label{eq:Gaussian_Approximation_of_BSCPermC}
		\log M^*(n,\epsilon) = \log \left( \Bigg \lceil \frac{(1-2 \delta)\sqrt{n}}{-\Phi^{-1}\left(\frac{\epsilon}{2}\right)} \Bigg \rceil \right)  + \theta
	\end{equation}
	is achievable  for all $n \ge N_0$, where $\theta$ is a constant.
\end{theorem}
\begin{proof}
	Instead of using (\ref{eq:Gamma_message_set_DMCs}), we use (\ref{eq:Gamma_message_set_binary}) with $\Delta_{1}^* = \{ (p,1-p): \delta \le p \le 1-\delta\}$. Repeat the argument of the proof of Theorem \ref{thm:Approximation_of_DMCs} replacing Lemma \ref{lemma:properities_Vn_Tn} with Lemma \ref{lemma:properities_Vn_Tn_binary}. Note that for  $\mathcal{N}_{r_m,2}^*$ constructed by (\ref{eq:Gamma_message_set_binary}), we have
	\begin{equation}
		r = \frac{1}{1-2\delta} \sqrt{\frac{r_m}{2 \log e}}.
	\end{equation}
	Using Proposition \ref{proposition:packing_num_binary} to continue as follows:
	\begin{align}
		\log M^*(n,\epsilon) 
		& \ge \log \left( \Bigg \lfloor \frac{(1-2 \delta)\sqrt{n}}{-\Phi^{-1}\left(\frac{\epsilon}{2}\right)} \Bigg  \rfloor + 1 \right)  + \log F_0 \nonumber \\
		& \ge \log \left( \Bigg \lceil \frac{(1-2 \delta)\sqrt{n}}{-\Phi^{-1}\left(\frac{\epsilon}{2}\right)} \Bigg \rceil \right)  + \log F_0,
	\end{align}
	where $F_0 \ge 0$ is a constant. This completes the proof.
\end{proof} 

Next, for BEC permutation channels, we have the following approximation. 

\begin{theorem}\label{thm:Gaussian_Approximation_of_BECPermC}
	For BEC permutation channels with erasure probability $\eta$, there exists a number $N_0 \ge 1$, such that
	\begin{equation}
		\log M^*(n,\epsilon)  = \log \left( \frac{(1-\eta)\sqrt{n}}{-\Phi^{-1}\left(\frac{\epsilon}{2}\right)} \right)  + \theta
	\end{equation}
	and
	\begin{equation}
		\log M^*(n,\epsilon) = \log \left( \Bigg \lceil \frac{(1-\eta)\sqrt{n}}{-\Phi^{-1}\left(\frac{\epsilon}{2}\right)} \Bigg \rceil \right)  + \theta
	\end{equation}
	are achievable for all $n \ge N_0$, where $\theta$ is a constant.
\end{theorem}
\begin{proof}
	Through Theorem \ref{thm:random_coding_bound_BEC}, repeat the argument of the proof of Corollary \ref{corollary:Gaussian_Approximation_of_BSCPermC_weak} and Theorem \ref{thm:Gaussian_Approximation_of_BSCPermC}, replacing $\delta$ with $\eta / 2$.
\end{proof}

	Finally, we remark that the Gaussian approximation shows some properties of the maximal code size with a given blocklength $n$ and probability of error $\epsilon$. In BSC permutation channels, while the channel capacity is only related to the rank of the channel matrix, the rate at which the achievable maximal code size approaches the capacity is affected by crossover probability $\delta$. 

\subsection{Numerical Results}

Here, we give the numerical results. Since the BEC permutation channel with erasure probability $\eta$ can be equivalent to the BSC permutation channel with crossover probability $\eta/2$ according to Theorem \ref{thm:random_coding_bound_BEC}, we focus on the numerical results of BSC permutation channels. We use Theorem \ref{thm:random_coding_bound_BSC_Perm} to compute the non-asymptotic achievability bound. We start searching from $M=2$ until the right side of (\ref{eq:thm:random_coding_bound_BSC_Perm_1}) is less than the error probability $\epsilon$. For Gaussian approximation, we use (\ref{eq:Gaussian_Approximation_of_BSCPermC_weak}) and (\ref{eq:Gaussian_Approximation_of_BSCPermC}) but omit the remainder term $\theta$. By comparison, we see that although the remainder term of the Gaussian approximation is a constant, it is still quite close to the non-asymptotically achievability bound. In fact, for all $n \ge 20$, the difference between (\ref{eq:Gaussian_Approximation_of_BSCPermC}) and Theorem \ref{thm:random_coding_bound_BSC_Perm} is within $1$ bits in $\log M^*(n,\epsilon)$. 
\begin{figure}[t]
	\centering
	\includegraphics[width = 0.4\textwidth]{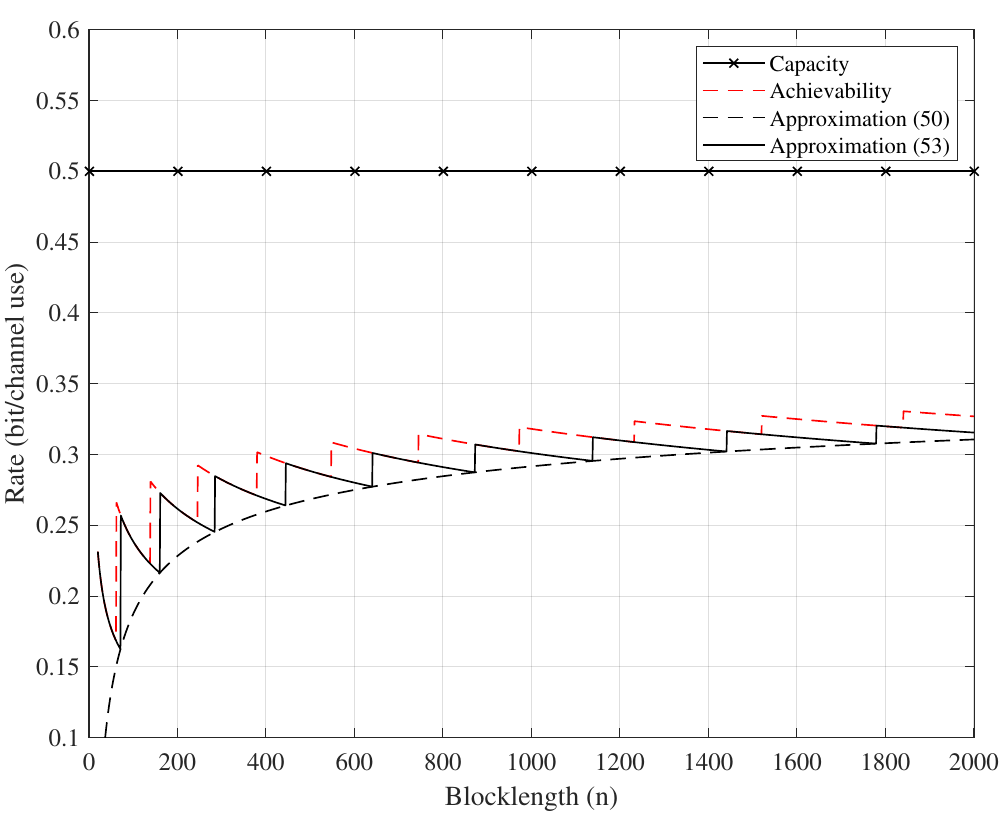}
	\captionsetup{font=footnotesize}
	\caption{\  Rate-blocklength tradeoff for the BSC with crossover probability $\delta=0.11$ and average block error rate $\epsilon = 10^{-3}$: Gaussian approximation}
	\label{fig:BSC_11_e-3_approximation}
\end{figure}

\begin{figure}[t]
	\centering
	\includegraphics[width = 0.4\textwidth]{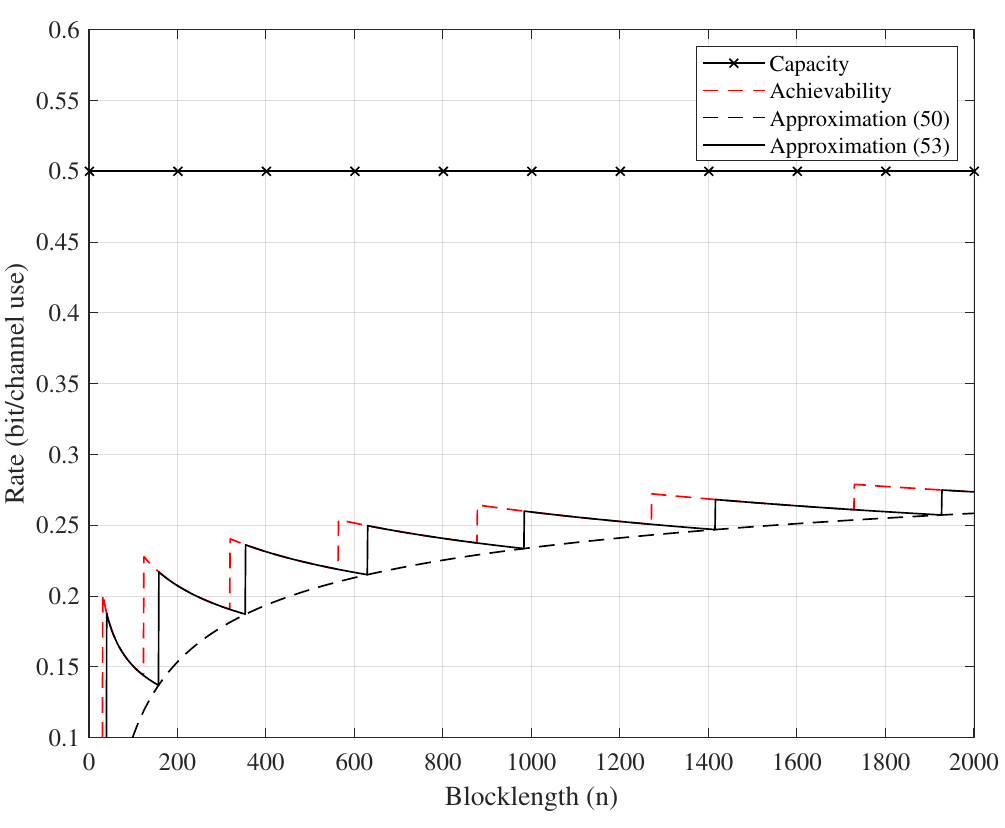}
	\captionsetup{font=footnotesize}
	\caption{\  Rate-blocklength tradeoff for the BSC with crossover probability $\delta=0.11$ and average block error rate $\epsilon = 10^{-6}$: Gaussian approximation}
	\label{fig:BSC_11_e-6_approximation}
\end{figure}

\begin{figure}[t]
	\centering
	\includegraphics[width = 0.4\textwidth]{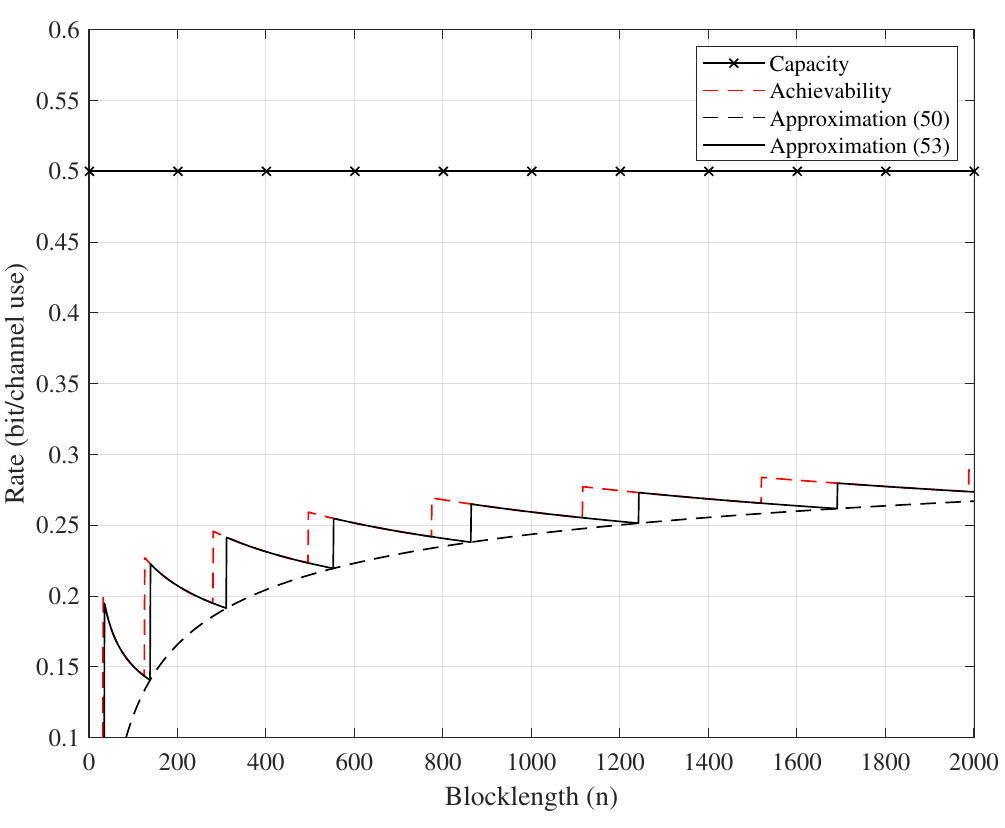}
	\captionsetup{font=footnotesize}
	\caption{\  Rate-blocklength tradeoff for the BSC with crossover probability $\delta=0.22$ and average block error rate $\epsilon = 10^{-3}$: Gaussian approximation}
	\label{fig:BSC_22_e-3_approximation}
\end{figure}

\begin{figure}[t]
	\centering
	\includegraphics[width = 0.4\textwidth]{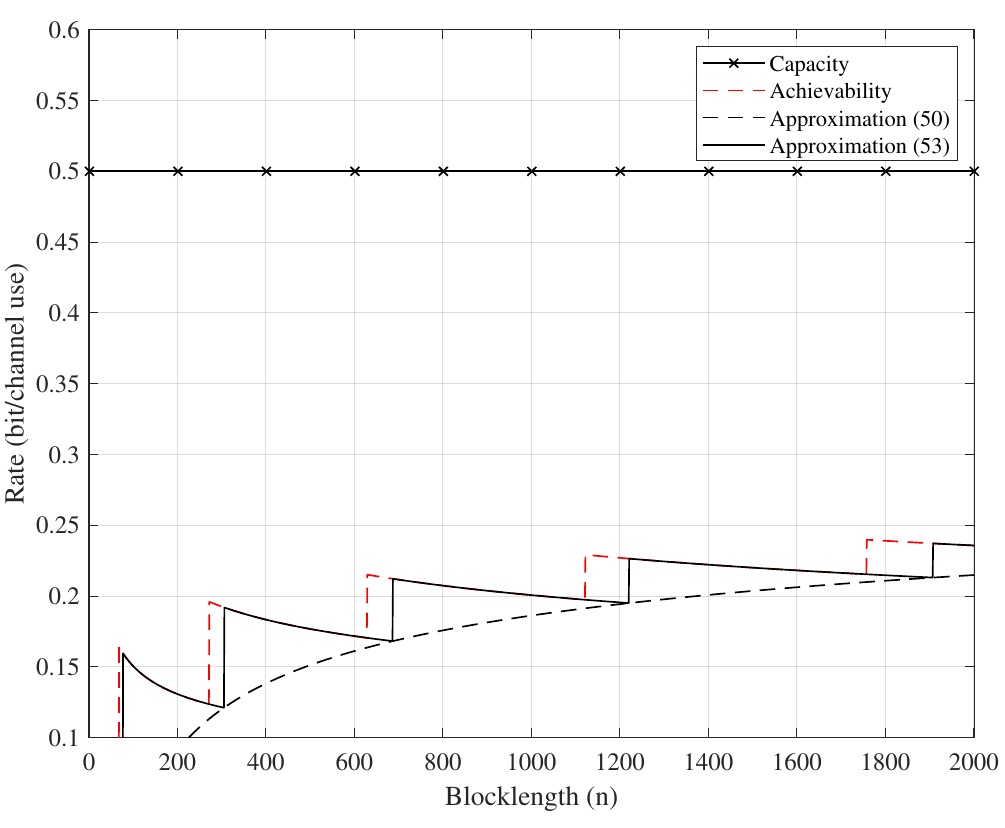}
	\captionsetup{font=footnotesize}
	\caption{\  Rate-blocklength tradeoff for the BSC with crossover probability $\delta=0.22$ and average block error rate $\epsilon = 10^{-6}$: Gaussian approximation}
	\label{fig:BSC_22_e-6_approximation}
\end{figure}

	Additionally, we compare our bound with the previous bound. To reduce the complexity, we use approximation (\ref{eq:Gaussian_Approximation_of_BSCPermC}) to compare with Makur's achievability bound for BSC permutation channels in \cite{makur_coding_2020}. The results show that our new achievability bound is uniformly better than Makur's bound.  In the setup of Figure \ref{fig:BSC_11_e-3_approximation}, our bound quickly approaches half of the capacity ($n \approx 300$). As the blocklength increases, Makur's bound reaches $20\%$ of the channel capacity at about $n \approx 5\times 10^4$, shown in Figure \ref{fig:Compare_achievability}. There are two reasons for this. First, we use the optimal ML decoder. Second, we show a relationship between error events, which reduces the number of error events when applying the union bound. On the other hand, Theorem \ref{thm:random_coding_bound}, in contrast to Makur's bound, yields a suitable asymptotic expansion of $\log M$.

\begin{figure}[t]
	\centering
	\includegraphics[width = 0.4\textwidth]{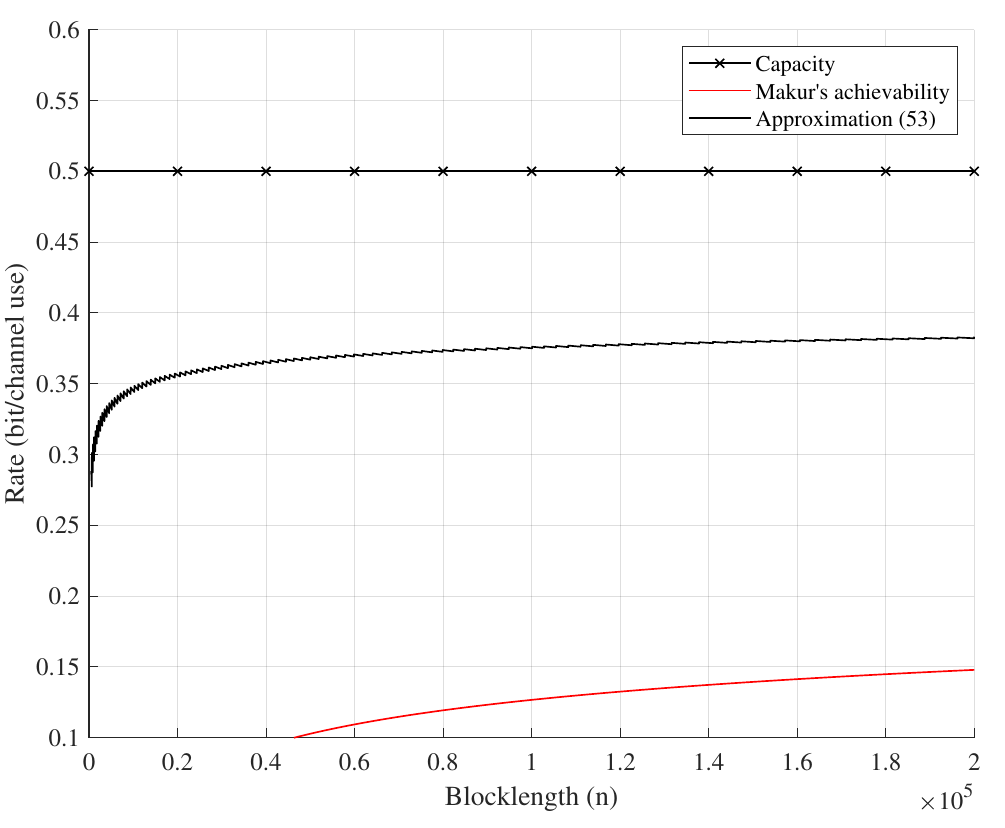}
	\captionsetup{font=footnotesize}
	\caption{\  Rate-blocklength tradeoff for the BSC with crossover probability $\delta=0.11$ and average block error rate $\epsilon = 10^{-3}$:comparison of the bounds.}
	\label{fig:Compare_achievability}
\end{figure}

\section{Conclusion and Discussion}\label{Section_conclusion}

	In summary, we established a new achievability bound for noisy permutation channels with a strictly positive and full-rank square matrix. The key element is that our analysis shows that the number of error events in the union is independent of the number of messages. This allows us to derive a refined asymptotic analysis of the achievable rate. Numerical simulations show that our new achievability bound is stronger than Makur's achievability bound in \cite{makur_coding_2020}. Additionally, our approximation is quite accurate, although the remainder term is constant. Finally, the primary future work will generalize the non-full rank and non-strictly positive matrices analysis. Other future work may improve asymptotic expansion (e.g., improving the remainder term to $o(1)$).

\appendices

\section{Divergence Packing} \label{Appendix_divergence_packing}

We first establish the basic result of divergence packing on $\Delta_{|\mathcal{Y}|-1}$. Consider the set 
\begin{align} 
	\Gamma_{r,|\mathcal{Y}|} = \bigg \{  P \in \Delta_{|\mathcal{Y}|-1}: & P = \left( \frac{a_1}{\lfloor 1/r \rfloor},...,\frac{a_{|\mathcal{Y}|}}{\lfloor 1/r \rfloor}\right) \nonumber \\
	& \ \ {\rm where} \ a_1,...,a_{|\mathcal{Y}|} \in \mathbb{Z}_{\ge 0}  \bigg \}.
\end{align} 

We have the following lemma:
\begin{lemma} \label{lemma:TV_radius}
	For $P,Q \in \Gamma_{r,|\mathcal{Y}|}$, we have
	\begin{equation}
		\min \limits_{P \neq Q} \TV(P,Q) \ge r.
	\end{equation}
\end{lemma}
\begin{proof}
	For every $P \in \Gamma_{r,|\mathcal{Y}|}$, we can find a $Q \in \Gamma_{r,|\mathcal{Y}|}$ such that for $m,n \in \mathcal{Y}$ and $m \neq n$, we have $q_m \neq p_m$, $q_n \neq p_n$. For $i \in \mathcal{Y} \setminus \{ m,n \}$, we have $q_i = p_i$. Then, we obtain that
	$\min \limits_{P \neq Q} \TV(P,Q) 
	= \min \limits_{P \neq Q} \frac{1}{2} \sum_{i \in \mathcal{Y}} |p_i-q_i| 
	= \frac{1}{\lfloor 1/r \rfloor} \ge r$.
\end{proof}
Then, we give the lower bound of the divergence packing number. 
\begin{lemma} \label{lemma:packing_num_full}
	Fix a divergence packing radius with $r_0$. We can find a $ \mathcal{N}_{r_0,|\mathcal{Y}|} = \Gamma_{r,|\mathcal{Y}|}$ such that 
	\begin{equation}
		|\mathcal{N}_{r_0,|\mathcal{Y}|}| = \binom{\lfloor 1/r \rfloor + |\mathcal{Y}|- 1}{|\mathcal{Y}|-1},
	\end{equation}
	and thus
	\begin{equation}
		|\mathcal{N}_{r_0,|\mathcal{Y}|}| \ge \left( \frac{1/r + |\mathcal{Y}| - 2}{|\mathcal{Y}|-1} \right)^{|\mathcal{Y}|-1},
	\end{equation}
	where $r = \sqrt{\frac{r_0}{2 \log e}}$.
\end{lemma}
\begin{proof}
	Fix a radius $r = \sqrt{\frac{r_0}{2 \log e}}>0$. For any $P,Q \in \Gamma_{r,|\mathcal{Y}|}$, we note that by Pinsker's inequality \cite[Theorem 7.10]{polyanskiy_2023_book} and Lemma \ref{lemma:TV_radius},
	\begin{equation}
		\min \limits_{P \neq Q}  D(P\|Q) \ge \min \limits_{P \neq Q} (2 \log e) \TV ^2 (P,Q) \ge r_0.
	\end{equation}
	Hence, any total variation packing with radius $\sqrt{\frac{r_0}{2 \log e}}$ must also give a divergence packing with radius $r_0$. 
	By simple application of combinatorial mathematics (e.g., see \cite[Chapter 1]{flajolet2009analytic}), we have
	\begin{equation}
		|\Gamma_{r,|\mathcal{Y}|}| = \binom{\lfloor 1/r \rfloor + |\mathcal{Y}| - 1}{|\mathcal{Y}|-1}.
	\end{equation}
	Note that the binomial coefficient have the following upper and lower bound (e.g., see \cite[Eq. (2.49)]{Jennifer_2022}):
	\begin{equation}
		\left( \frac{n}{k}\right)^{k} \le \binom{n}{k} \le \left(\frac{ne}{k}\right)^{k}.
	\end{equation}
	We have
	\begin{equation}
		 |\Gamma_{r,|\mathcal{Y}|}|  \ge \left( \frac{\lfloor 1/r \rfloor + |\mathcal{Y}| - 1}{|\mathcal{Y}|-1} \right)^{|\mathcal{Y}|-1}.
	\end{equation}
	Using $\lfloor 1/r \rfloor \ge 1/r-1$, we obtain
	\begin{equation}
		|\Gamma_{r,|\mathcal{Y}|}| \ge \left( \frac{1/r + |\mathcal{Y}| - 2}{|\mathcal{Y}|-1} \right)^{|\mathcal{Y}|-1}.
	\end{equation}
	Let $\mathcal{N}_{r_0,|\mathcal{Y}|} = \Gamma_{r,|\mathcal{Y}|}$ and note that $r = \sqrt{\frac{r_0}{2 \log e}}$, we complete the proof.
\end{proof}

Now, we can prove Theorem \ref{thm:packing_num_subs} and Proposition \ref{proposition:packing_num_binary}.

{\noindent \textit{Proof of Theorem \ref{thm:packing_num_subs}.} }
	We first note that $\Delta_{|\mathcal{Y}|-1}^*$ is convex because the set of input distribution $\Delta_{|\mathcal{X}|-1}$ is convex. Since the grid structure $\Gamma_{r,|\mathcal{Y}|}$ is uniform on $\Delta_{|\mathcal{Y}|-1}$, we only need to consider the loss near a $(|\mathcal{Y}|-2)$-dimensional probability space $\mathcal{B}$ on $\Delta_{|\mathcal{Y}|-1}$, in which space $\mathcal{B}$ is the boundary of $\Delta_{|\mathcal{Y}|-1}^*$. The set of possible losing points is $\mathcal{P} = \{P \in \mathcal{B}: \min_{P,Q} \TV(P\|Q) < \frac{1}{\lfloor 1/r \rfloor}, Q \in \Gamma_{r,|\mathcal{Y}|} \}$. By counting, we have the upper bound 
	\begin{align}
		|\mathcal{P}| 
		& \le 2|\mathcal{Y}|\binom{\lfloor 1/r \rfloor + |\mathcal{Y}| - 2}{|\mathcal{Y}| - 2} \nonumber \\
		& \le 2 |\mathcal{Y}| \left( \frac{( 1/r + |\mathcal{Y}| - 2)e}{|\mathcal{Y}| - 2} \right)^{|\mathcal{Y}| - 2}.
	\end{align}
	Then, we use the volume ratio to obtain the result. 
{\hfill $\square$}

{\noindent \textit{Proof of Proposition \ref{proposition:packing_num_binary}.} }
	We note that $\Gamma_{r,2}^*$ in (\ref{eq:Gamma_message_set_binary}) is a subset of $\Delta_{1}^*$. Fix $r = \frac{1}{\xi} \sqrt{\frac{r_0}{2 \log e}}$. Use Lemma \ref{lemma:TV_radius} and Pinsker's inequality \cite[Theorem 7.10]{polyanskiy_2023_book} to obtain $\min \limits_{P \neq Q}  D(P\|Q) \ge r_0$. Then, we conclude this proof by realizing that $|\Gamma_{r,2}^*| = \lfloor 1/r \rfloor +1$.
{\hfill $\square$}

\section{Proof of Lemma \ref{lemma:event_disjoint}}	\label{Appendix_lemma_event_disjoint}

We start by a lemma:
\begin{lemma} \label{lemma:event_disjoint_pre}
	Let $\mathcal{N}_{r_0,|\mathcal{Y}|}^*$ be constructed by (\ref{eq:Gamma_message_set_DMCs}) or (\ref{eq:Gamma_message_set_binary}). Assuming the transmitted message is $m$ corresponding to the marginal distribution $P$. Then, for every $\Lambda \in \Delta_{|\mathcal{Y}|-1}$ and $Q = (q_1,...,q_{|\mathcal{Y}|}) \in \mathcal{N}_{r_0,|\mathcal{Y}|}^* \setminus \mathcal{R}_{P} \cap \{ P\}$, if
	\begin{equation}\label{eq:lemma:event_disjoint_1}
		\prod_{i=1}^{|\mathcal{Y}|} p_i^{\lambda_i} \le \prod_{i=1}^{|\mathcal{Y}|} q_i^{\lambda_i},
	\end{equation}
	there exists a $Q^* =(q^*_1,...,q^*_{|\mathcal{Y}|}) \in \mathcal{R}_{P}$ such that
	\begin{equation}\label{eq:lemma:event_disjoint_2}
		\prod_{i=1}^{|\mathcal{Y}|} p_i^{\lambda_i} \le \prod_{i=1}^{|\mathcal{Y}|} (q_i^*)^{\lambda_i}.
	\end{equation}
\end{lemma}
\begin{proof}
	We first consider $\mathcal{N}_{r_0,|\mathcal{Y}|}^*$ constructed by (\ref{eq:Gamma_message_set_DMCs}).  Fix a $P \in \mathcal{N}_{r_0,|\mathcal{Y}|}^*$ and generate $\mathcal{R}_{P}$ corresponding to $P$. Fix $m,n \in \mathcal{Y}$ and $m \neq n$. Denote a $Q^*_{m,n}$ such that $q^*_m = p_m + \frac{1}{\lfloor 1/r \rfloor}$, $q^*_n = p_n - \frac{1}{\lfloor 1/r \rfloor}$, $q^*_i = p_i$, where $i \in \mathcal{Y}\setminus \{m,n\}$. Let
	\begin{equation}
		\left| \log \frac{p_{m}}{p_{m} + 1/\lfloor 1/r \rfloor} \right| = G_{m,n} \left| \log \frac{p_n}{p_n - 1/\lfloor 1/r \rfloor} \right|,
	\end{equation}
	where $G_{m,n}$ is a constant. Define
	\begin{equation}
		\mathcal{B}_{m,n} = \left\{  (\lambda_1,...,\lambda_{|\mathcal{Y}|}):  G_{m,n} \lambda_m < \lambda_n \right\}.
	\end{equation}
	If $ Q^*_{m,n} \in \mathcal{R}_P$, for $\Lambda = (\lambda_1,...,\lambda_{|\mathcal{Y}|}) \in \mathcal{B}_{m,n}$ and $Q^*_{m,n} = (q^*_1,...,q^*_{|\mathcal{Y}|})$, we have 
	\begin{equation} \label{eq:lemma_disjoint_pre_proof_1}
			\prod_{i=1}^{|\mathcal{Y}|} p_i^{\lambda_i} > \prod_{i=1}^{|\mathcal{Y}|} (q_i^*)^{\lambda_i}.
	\end{equation}
	
	For $m,n \in \mathcal{Y}$ and $m \neq n$, we find $\mathcal{B}_{m,n}$ and consider the intersection, i.e., $\mathcal{B}_P = \bigcap_{m,n \in \mathcal{Y}, m \neq n, Q^*_{m,n} \in \mathcal{R}_P} \mathcal{B}_{m,n} $. Then, for $\Lambda \in \mathcal{B}_P$, we have (\ref{eq:lemma_disjoint_pre_proof_1})
	holds for any $Q^* \in \mathcal{R}_{P}$. 
	
	Now we consider $Q\in \mathcal{N}_{r_0,|\mathcal{Y}|}^* \setminus \mathcal{R}_{P} \cap \{ P\}$, which is farther from $P$ in Euclidean distance than $Q^* \in \mathcal{R}_P$. For each $i \in \mathcal{Y}$, we assume that the distance $p_i$ along $q_i$ is $\frac{K_i}{\lfloor 1/r \rfloor}$, where $K_i \in \mathbb{Z}_{\ge 0}$.
	
	Since $p_i -  \frac{K_i}{\lfloor 1/r \rfloor} > 0$ and $K_i \ge 0$, we have
	\begin{equation}
		K_i \left| \log \frac{p_{i}}{p_{i} + 1/\lfloor 1/r \rfloor} \right| \ge \left| \log \frac{p_{i}}{p_{i} + K_i/\lfloor 1/r \rfloor } \right|
	\end{equation}
	and
	\begin{equation}
		K_i\left| \log \frac{p_{i}}{p_{i} - 1/\lfloor 1/r \rfloor} \right| \le \left| \log \frac{p_{i}}{p_{i} - K_i/\lfloor 1/r \rfloor} \right|.
	\end{equation}
	Let $\mathcal{Y}^* = \{ i \in \mathcal{Y}:q_i \neq p_i \}$,  $\mathcal{Y}_0 = \{ i \in \mathcal{Y}^*:q_i<p_i \}$, and $\mathcal{Y}_1 = \{ i \in \mathcal{Y}^*:q_i>p_i \}$. For $i \in \mathcal{Y}_0$, let $q_i' = p_{i} - 1/\lfloor 1/r \rfloor$. For $i \in \mathcal{Y}_1$, let $q_i' = p_{i} + 1/\lfloor 1/r \rfloor$. Then, we have
	\begin{equation} \label{eq:lemma_disjoint_pre_proof_3}
		\sum_{i \in \mathcal{Y}^*} \lambda_i \log \frac{p_i}{q_i} \ge \sum_{i \in \mathcal{Y}^*} \lambda_i K_i \log \frac{p_i}{q_i'}.
	\end{equation}
	Since the restriction of the probability space, we have 
	\begin{equation}\label{eq:lemma_disjoint_pre_proof_4}
	\sum_{i \in \mathcal{Y}_0} K_i = \sum_{i \in \mathcal{Y}_1} K_i.
	\end{equation}
	 Recal that for $ \Lambda \in \mathcal{B}_P$, (\ref{eq:lemma_disjoint_pre_proof_1}) holds for any $Q^* \in \mathcal{R}_{P}$. Then, by the definition of $\mathcal{R_P}$, for $\Lambda \in \mathcal{B}_P$, we have  
	\begin{equation}\label{eq:lemma_disjoint_pre_proof_5}
		\lambda_m \log \frac{p_m}{q_m'} + \lambda_n \log \frac{p_n}{q_n'} > 0,
	\end{equation}
	where $m \in \mathcal{Y}_0$ and $n \in \mathcal{Y}_1$. 
	We combine (\ref{eq:lemma_disjoint_pre_proof_3}), (\ref{eq:lemma_disjoint_pre_proof_4}), and (\ref{eq:lemma_disjoint_pre_proof_5}) to obtain that if  $\Lambda \in \mathcal{B}_P$, then 
	\begin{equation} \label{eq:lemma_disjoint_pre_proof_2}
		\sum_{i \in \mathcal{Y}^*} \lambda_i \log \frac{p_i}{q_i} > 0.
	\end{equation}
	That is
	\begin{equation}
				\prod_{i=1}^{|\mathcal{Y}|} p_i^{\lambda_i} > \prod_{i=1}^{|\mathcal{Y}|} q_i^{\lambda_i}.
	\end{equation}
	
	For $Q \in \mathcal{N}_{r_0,|\mathcal{Y}|}^* \setminus \mathcal{R}_{P} \cap \{ P\}$, define $\mathcal{B}_Q$ such that (\ref{eq:lemma_disjoint_pre_proof_2}) holds for $\Lambda \in \mathcal{B}_Q$. We have $\mathcal{B}_P \subset \mathcal{B}_Q$.
	Denote by $\mathcal{A}_Q = \Delta_{|\mathcal{Y}|-1} \setminus \mathcal{B}_Q$ and $\mathcal{A}_P = \Delta_{|\mathcal{Y}|-1} \setminus \mathcal{B}_P$ the complement of $\mathcal{B}_Q$ and $\mathcal{B}_P$, respectively. We note that if (\ref{eq:lemma:event_disjoint_1}) holds, we have $\Lambda \in \mathcal{A}_Q$. Then, we obtain (\ref{eq:lemma:event_disjoint_2}) holds for a $Q^* \in \mathcal{R}_{P}$ since we obviously have $\mathcal{A}_Q \subset \mathcal{A}_P$. 
	This completes the proof in the case of $\mathcal{N}_{r_0,|\mathcal{Y}|}^*$ constructed by (\ref{eq:Gamma_message_set_DMCs}).
	
	We then give another proof in binary case. Consider $\mathcal{N}_{r_0,|\mathcal{Y}|}^*$ constructed by (\ref{eq:Gamma_message_set_binary}). For convenience, let $p_j < p_m$ if $j < m$. Fix $P_{m-1} = (p_{m-1},1-p_{m-1})$ and $P_{m} = (p_m,1-p_m)$. Let 
	\begin{equation} \label{eq:binary_case_event_compute}
		f(p_j) =  \left( \log \frac{1-p_{j}}{1-p_m}\right) \bigg / \left( \log \frac{p_m(1-p_{j})}{p_{j}(1-p_{m})} \right).
	\end{equation}
	For any $\lambda \in [0,1]$ and $j \in [m-2]$, we have
	\begin{equation} \label{eq:binary_case_event_3}
		p_{j}^{\lambda}(1-p_{j})^{(1-\lambda)} \ge p_{m}^{\lambda}(1-p_{m})^{(1-\lambda) }
	\end{equation}
	holds for $\lambda \in [0,f(p_j)]$.
	Note that $f(p_j)$ is a monotonically increasing function with respect to $p_{j}$. Then, we can obtain 
	\begin{equation} \label{eq:binary_case_event_1}
			p_{m-1}^{\lambda}(1-p_{m-1})^{(1-\lambda)} \ge p_{m}^{\lambda}(1-p_{m})^{(1-\lambda) }
		\end{equation}
	since $[0,f(p_j)] \subset [0,f(p_{m-1})]$. 
	
	For $j \in \mathbb{Z}_{\ge m+2} \cap [m+2,|\mathcal{M}|]$, use the same argument to obtain that (\ref{eq:binary_case_event_3}) holds for $ \lambda \in [f(p_j),1] $ and \begin{equation} \label{eq:binary_case_event_2}
			p_{m+1}^{\lambda}(1-p_{m+1})^{(1-\lambda)} \ge p_{m}^{\lambda}(1-p_{m})^{(1-\lambda) }
	\end{equation}
	holds for $ \lambda \in [f(p_{m+1}),1] $. 
	Then, (\ref{eq:binary_case_event_3}) implies (\ref{eq:binary_case_event_2}) since $[f(p_j),1] \subset [f(p_{m-1}),1]$.
	This completes the proof.
\end{proof}

Now, we turn to prove Lemma \ref{lemma:event_disjoint}.

{\noindent \textit{Proof of Lemma \ref{lemma:event_disjoint}.} }

Using Lemma \ref{lemma:event_disjoint_pre}, for $j \in [|\mathcal{M}|]$, if $\left\{  P_{j}^n(y^n) \ge P_{m}^n(y^n) \right\}$ occurs, we obtain there must have $j \in [|\mathcal{R}_{P_m}|]$ such that $\left\{  Q ^n_j(y^n) \ge P_{m}^n(y^n) \right\}$ occurs. Then, we observe
\begin{align}
	& \mathbb{P} \left[\bigcup_{j=1,j\neq m}^{|\mathcal{M}|}\left\{  P_{j}^n \ge P_{m}^n \right\} \right] \nonumber \\
	&= \sum_{y^n \in \mathcal{Y}^n}  P_{m}^n(y^n) \bigcup_{j=1,j\neq m}^{|\mathcal{M}|}\left\{  P_{j}^n(y^n) \ge P_{m}^n(y^n) \right\} \label{eq:lemma:event_disjoint_proof_1} \\
	&= \sum_{y^n \in \mathcal{Y}^n}  P_{m}^n(y^n) \bigcup_{j =1}^{|\mathcal{R}_{P_m}|} \left\{  Q_j^n(y^n) \ge P_{m}^n(y^n) \right\}. \label{eq:lemma:event_disjoint_proof_2}  \\
	& = \mathbb{P} \left[ \bigcup_{j =1}^{|\mathcal{R}_{P_m}|} \left\{ Q_j^n \ge P_m^n \right\} \right],
\end{align}
where in (\ref{eq:lemma:event_disjoint_proof_1}) we sum over all possible outputs, (\ref{eq:lemma:event_disjoint_proof_2}) relies on Lemma \ref{lemma:event_disjoint_pre}.
This completes the proof of (\ref{eq:lemma:event_disjoint_high_dim}).
{\hfill $\square$}

\section{Properties of $V(P\|Q)$ and $T(P\|Q)$} \label{Appendix_lemma_properities_Vn_Tn}

This appendix is concerned with the behavior of $V(P\|Q)$ and $T(P\|Q)$. We first prove Lemma \ref{lemma:properities_Vn_Tn}.

{\noindent \textit{Proof of Lemma \ref{lemma:properities_Vn_Tn}.} }

	Since the DMC matrix is strictly positive, each term of marginal distribution $P$ is uniformly bounded away from zero, i.e., there exists a $p_{min} \in (0,1)$ such that $ P(y) \ge p_{min}$ for all $y \in \mathcal{Y}$. Similarly, there exists a $p_{max} \in (0,1)$ such that $ P(y) \le p_{max}$ for all $y \in \mathcal{Y}$. Clearly, we have $p_{min} = \min_{y \in \mathcal{Y}} \{ P(y):P \in \mathcal{N}_{r_0,|\mathcal{Y}|}^* \}$ and $p_{max} = \max_{y \in \mathcal{Y}} \{ P(y):P \in \Delta_{|\mathcal{Y}|-1}^* \} $. Without loss of generality, we assume $|\mathcal{N}_{r_0,|\mathcal{Y}|}^*| \ge 2$. That is, we have $r_0 \le \frac{2 \log e}{9}$.
	
	We know that 
	\begin{align}
		P(y) 
		& = \frac{a}{1/r - \delta} \nonumber \\
		& = \frac{a\sqrt{\frac{r_0}{2 \log e}}}{1-\delta \sqrt{\frac{r_0}{2 \log e}}},
	\end{align}
	where $\delta \in [0,1)$ and $a \in \mathbb{Z}_{\ge 1}$. 
	
	We consider each term $P(y)$ in $P$ separately. If $P(y) = Q(y)$, we obviously get zero. If $P(y) > Q(y)$, we have $Q(y) = \frac{(a-1)\sqrt{\frac{r_0}{2 \log e}}}{1-\delta \sqrt{\frac{r_0}{2 \log e}}}$. Here, since the distribution in $\mathcal{N}_{r_0,|\mathcal{Y}|}^*$ have no zero term, we have $a \in \mathbb{Z}_{\ge 2}$. Since $a = P(y)\left( 1-\delta \sqrt{\frac{r_0}{2 \log e}} \right)\sqrt{\frac{2\log e}{r_0}}$, we have
	\begin{align} \label{eq:lemma:properities_Vn_Tn_proof_1}
		&\frac{\sqrt{r_0\log e}}{P(y) (\sqrt{2}-\delta\sqrt{r_0/\log e})} \le \log \frac{P(y)}{Q(y)} \nonumber \\
		&\ \ \ \ \ \ \ \le \frac{\sqrt{r_0\log e}}{P(y) (\sqrt{2}-\delta\sqrt{r_0/\log e}) - \sqrt{r_0/\log e}}.
	\end{align}
	Consequently, we have
	\begin{equation} \label{eq:lemma:properities_Vn_Tn_proof_2}
		P(y) \left(\log \frac{P(y)}{Q(y)} \right)^2 \ge \frac{r_0 \log e}{2P(y)}
	\end{equation}
	and
	\begin{align} \label{eq:lemma:properities_Vn_Tn_proof_3}
		P(y) \left(\log \frac{P(y)}{Q(y)} \right)^2 
		&\le \frac{r_0 \log e}{2(a-1)^2 P(y)/a^2} \nonumber \\
		&\le \frac{2r_0 \log e}{P(y)} \left( \frac{1}{1-\delta\sqrt{\frac{r_0}{2 \log e}}} \right)^2.
	\end{align}
	
	If $P(y) < Q(y)$, consider $Q(y) = \frac{(a+1)\sqrt{\frac{r_0}{2 \log e}}}{1-\delta \sqrt{\frac{r_0}{2 \log e}}}$, where $a \in \mathbb{Z}_{\ge 1} $. Applying the same argument we have
	\begin{align}
		&\frac{\sqrt{r_0\log e}}{P(y) (\sqrt{2}-\delta\sqrt{r_0/\log e}) + \sqrt{r_0/\log e}} \le \bigg|  \log \frac{P(y)}{Q(y)} \bigg| \nonumber \\
		& \ \ \ \ \ \ \ \ \ \ \le \frac{\sqrt{r_0\log e}}{P(y) (\sqrt{2}-\delta\sqrt{r_0/\log e})}.
	\end{align}
	Consequently, we have
	\begin{align} \label{eq:lemma:properities_Vn_Tn_proof_4}
		  P(y) \left(\log \frac{P(y)}{Q(y)} \right)^2 
		  & \ge \frac{r_0 \log e}{2(a+1)^2 P(y)/a^2} \nonumber \\
		  &\ge \frac{r_0 \log e}{8P(y)}
	\end{align}
	and
	\begin{align}\label{eq:lemma:properities_Vn_Tn_proof_5}
	 	P(y) \left(\log \frac{P(y)}{Q(y)} \right)^2 
	 	& \le \frac{r_0 \log e}{2P(y)} \left( \frac{1}{1-\delta\sqrt{\frac{r_0}{2 \log e}}} \right)^2.
	\end{align}
	
	Then, we see that
	\begin{align}
		V(P\|Q) &\ge - D(P\|Q)^2 + \sum_{y} P(y) \left(\log \frac{P(y)}{Q(y)} \right)^2 \label{eq:lemma:properities_Vn_Tn_proof_6}  \\
		& \ge r_0 \log e  \left( \frac{5}{8p_{max}} - \frac{r_0}{\log e} \right), \label{eq:lemma:properities_Vn_Tn_proof_7}
	\end{align}
	where (\ref{eq:lemma:properities_Vn_Tn_proof_6}) simply expand the variance, (\ref{eq:lemma:properities_Vn_Tn_proof_7}) holds  since the definition of $\mathcal{R}_{P}$ ensures that for all $y\in \mathcal{Y}$, we only have two terms that are not zero and can be bounded by (\ref{eq:lemma:properities_Vn_Tn_proof_2}) and (\ref{eq:lemma:properities_Vn_Tn_proof_4}), respectively. Finally, we complete the proof of the lower bound by noting that $r_0 \le \frac{2 \log e}{9}$.
	
	Note that
	\begin{equation}
		\left( \frac{1}{1-\delta\sqrt{\frac{r_0}{2 \log e}}} \right)^2 \le \frac{1}{(1-p_{max})^2}.
	\end{equation} 
	Using (\ref{eq:lemma:properities_Vn_Tn_proof_3}) and (\ref{eq:lemma:properities_Vn_Tn_proof_5}), we obtain the upper bound
	\begin{align}
		V(P\|Q) 
		& \le \frac{5 \log e}{2p_{min}(1-p_{max})^2} r_0, \label{eq:lemma:properities_Vn_Tn_proof_8}
	\end{align}
	which yields (\ref{eq:lemma:properities_Vn_Tn_1}).
	
	We now turn to bound $T(P\|Q)$. For $P(y) > Q(y)$, we have
	\begin{equation}
		P(y) \bigg | \log \frac{P(y)}{Q(y)} \bigg |^3 \le \frac{4}{\sqrt{2}}\frac{(r_0 \log e)^{3/2}}{P^2(y)} \left( \frac{1}{1-\delta\sqrt{\frac{r_0}{2 \log e}}} \right)^3.
	\end{equation}
	For $P(y) < Q(y)$, we have
	\begin{equation}
		P(y) \bigg |\log \frac{P(y)}{Q(y)} \bigg |^3 \le \frac{1}{2\sqrt{2}}\frac{(r_0 \log e)^{3/2}}{P^2(y)} \left( \frac{1}{1-\delta\sqrt{\frac{r_0}{2 \log e}}} \right)^3.
	\end{equation}
	Note that $D(P\|Q)$ can be bounded by the following: 
	\begin{equation}
		D(P\|Q) \le \frac{3\sqrt{2 \log e}}{2} \sqrt{r_0}  \left( \frac{1}{1-\delta\sqrt{\frac{r_0}{2 \log e}}} \right).
	\end{equation}
	Using the inequality $|a-b|^3 \le 4(|a|^3+|b|^3)$, we obtain 
	\begin{align}
		T(P\|Q) 
		& \le \frac{36\sqrt{2} (\log e )^{3/2}}{p_{min}^2 (1-p_{max})^3}  r_0^{3/2}, \label{eq:lemma:properities_Vn_Tn_proof_20}
	\end{align}
	which establishes (\ref{eq:lemma:properities_Vn_Tn_2}).
{\hfill $\square$}

Then, we prove Lemma \ref{lemma:properities_Vn_Tn_binary} for $\mathcal{N}_{r_0,2}^*$ constructed by (\ref{eq:Gamma_message_set_binary}).

{\noindent \textit{Proof of Lemma \ref{lemma:properities_Vn_Tn_binary}.} }
	
	We first note that
	\begin{align}
		P(y) 
		& = \frac{a \xi }{1/r - \delta} \nonumber \\
		& = \frac{\left( a - \frac{\delta_1\delta}{\xi} + F_0 \right)\sqrt{\frac{r_0}{2 \log e}}}{1-\frac{\delta}{\xi} \sqrt{\frac{r_0}{2 \log e}}},	  
	\end{align}
	where $\delta \in [0,1]$, $F_0 = \delta_1/\sqrt{\frac{r_0}{2 \log e}}$, and $a \in \mathbb{Z}_{\ge 1}$. Note that $F_0 - \frac{\delta_1\delta}{\xi} \ge 0$ holds for small $r_0$. Repeat the proof of Lemma \ref{lemma:properities_Vn_Tn}, replacing $a$ with $ a - \frac{\delta_1\delta}{\xi} + F_0 $ and $\delta$ with $\frac{\delta}{\xi}$.  Then, there exists a packing radius $r_1$, such that for all $r_0 \le r_1$, we have
	\begin{equation}
		F_0 r_0 \le V(P\|Q) \le F_1 r_0
	\end{equation}
	and 
	\begin{equation}
		T(P\|Q) \le F_2 r_0^{3/2},
	\end{equation}
	where $F_0$, $F_1$ and $F_2$ are positive and finite.
{\hfill $\square$}

\newpage

\printbibliography[heading=bibintoc, title={References}]


\newpage

\end{document}